\documentclass[10pt]{article}

\newenvironment{packed_item}{
\begin{itemize}
  \setlength{\itemsep}{0.5pt}
  \setlength{\parskip}{0pt}
  \setlength{\parsep}{0pt}
}{\end{itemize}}

\newenvironment{packed_enum}{
\begin{enumerate}
  \setlength{\itemsep}{0.5pt}
  \setlength{\parskip}{0pt}
  \setlength{\parsep}{0pt}
}{\end{enumerate}}

\usepackage[margin=1in]{geometry}
\usepackage{appendix}
\usepackage{amsfonts,amsmath,amssymb,amsthm}
\usepackage{url}
\usepackage{algorithm,algorithmic}
\usepackage{etoolbox}%\AtBeginEnvironment{algorithmic}{\small}
\usepackage{graphicx,psfrag}
\usepackage[usenames,dvipsnames]{color}

 \newtheorem{definition}{Definition}
 \newtheorem{theorem}{Theorem}
 \newtheorem{lemma}{Lemma}
 
 \newtheorem{corollary}{Corollary}

\newcommand{\Alg}{\mathcal{A}}
\newcommand{\eps}{\varepsilon}
\newcommand{\tree}{\mathcal{T}}
\newcommand{\E}{\mathbb{E}}
\renewcommand{\vec}[1]{\mathbf{#1}}
\newcommand{\uni}{\mathcal{U}}

\newcommand{\junk}[1]{}

\definecolor{red}{RGB}{255,0,0}
\newcommand{\red}{\textcolor{red}}
\definecolor{green}{RGB}{0,255,0}
\newcommand{\green}{\textcolor{green}}
\definecolor{blue}{RGB}{0,0,255}
\newcommand{\blue}{\textcolor{blue}}

\DeclareMathOperator{\Lap}{Lap}

\begin{document}
%\begin{titlepage}
\title{Private Sums on Decayed Streams}
\author{Jean Bolot \and Nadia Fawaz \and
  S. Muthukrishnan \and Aleksandar Nikolov \and Nina
  Taft}

\maketitle

\begin{abstract}
  In  monitoring applications, recent data is more important than distant data. How does this affect privacy of data analysis? We study a general class of data analyses --- computing predicate sums ---  with privacy.  Formally, we study the problem of estimating predicate sums {\em privately}, for sliding windows (and other well-known decay models of data, i.e.~exponential and polynomial decay). We extend the recently proposed continual privacy model of Dwork et al.~\cite{dwork-continual}.
% who initiated the study of differential privacy as data is continually updated over time, but without % the concept of data decay. In particular, 
%While we require accuracy in analysis with respect to the window, we still want differential privacy for the entire past. This is %challenging because window sums are not monotonic or even near-monotonic as the problems studied in~\cite{dwork-continual}. 
%Furthermore, the algorithms in~\cite{dwork-continual} have error that grows unbounded as more %updates are processed, while we require that our algorithms achieve error that is negligible with %respect to the bounded range of window sums. 
We present algorithms for decayed sum which are $\eps$-differentially private, and are accurate. For window and exponential decay sums, our algorithms are accurate up to additive $1/\eps$ and polylog terms in the range of the computed function; for polynomial decay sums which are technically more challenging because partial solutions do not compose easily, our algorithms incur additional relative error. Further, we show lower bounds, tight within polylog factors and tight with respect to the dependence on the probability of error. 
%Our results are obtained via a natural dyadic tree we maintain, but the crux is  we treat the tree data structure in 
% non-uniform manner.

\end{abstract}

%\end{titlepage}
%\setcounter{page}{1}
%\newpage
\section{Introduction}

Any nontrivial physical, hardware or software system has a dashboard continually observing the system variables, and updating various measurements.  In such applications, data arrives over time, and we need to continually output the result of some analysis $f$ 
%on data  $D_j$ seen thus far, 
for each time instant $j$ on all data seen thus far.  This challenges privacy of analysis because the same function is computed on several deltas of the data and the collection of these function values can potentially leak information.
% than several queries to different portions of static data.  
Recently, the notion of differential privacy was adopted to address this challenge~\cite{dwork-continual,chan2010private}, and we extend that study.
% here.

%In particular, recent work on differential privacy with continual observation and analysis
~\cite{dwork-continual,chan2010private} identified the problem of computing the running sum of a series of $0/1$ updates as an important technical primitive, formulated differential privacy of computing these running sums, and presented upper and lower bounds on accuracy of $\eps$-differentially private algorithms for computing running sums. They showed that an additive accuracy of $O(\frac{1}{\eps} \log^{2}  T)$ with constant probability is possible for the running sums problem, and that $\Omega(\log T)$ additive error was necessary to answer privately all running sum queries for all time steps $j \in [1, T]$. 

\medskip
\noindent
{\bf Power of Running Sums.} The sums problem 
%a rich problem
captures many analyses by  applying a suitable predicate  to the data items that map them to $0/1$. 
For example, at time $j$, say data item $D_j=(u_j,m_j)$ is the user ID $u_j$ and the name of the movie $m_j$ watched in an online service by $u_j$ at that time. A natural predicate is ${\cal P}_m(D_j)=1$ if $m_j=m$ and $0$ otherwise; the running sum with this predicate counts the number of user IDs that watched a particular film $m$. Another natural predicate is 
${\cal P}_u(D_j)=1$ if $u_j=u$ and $0$ otherwise; this running sum counts the number of movies watched by a user $u$.  %Similarly, many other predicates can be used to filter the data and compute COUNTs.  
The predicates can be different for different items. E.g., 
${\cal P}_{j,u}(D_j)=1$ if $u_j=u$ and $j \in [9,17]$ will filter movies watched by user $u$ during business hours $9$ AM to $5$ PM.  Even more generally, ${\cal P}$ may be 
a machine learning based classification routine such as whether a
click by any user from a certain IP address on an Internet ad is a
spam or not, and the running sum will count the total number of spam clicks from the given IP address. \hfill$\blacksquare$

\medskip
Our point of departure from prior work is that in reality, monitoring applications emphasize recent data more than data long past.  For example, monitoring applications typically consider a  ``window''' of continual observations such as, last $T$ time units, or last $W$ updates.  More generally, they discount items based on how far they are in the past, and analyze decayed data.
The commonly useful decay models  are exponential and polynomial 
decays~\cite{datar2002maintaining,cohen2003maintaining}. 

%  In the data streams literature, there is significant amount of work in defining various window % models
%The commonly useful decay models of exponential and polynomial decays. 
\junk{
and algorithms for various data stream decay  models~\cite{cohen2003maintaining}. See the papers on data streaming over sliding windows~\cite{datar2002maintaining} for motivation and applications of the window model and~\cite{cohen2003maintaining} for motivation of the exponential decay model as well as discussion of the usefulness and motivation of considering a polynomial decay model, as well as applications for both models.
}

\medskip
\noindent
{\bf Our results.}
Motivated by this, we consider differential privacy of continual observations over windows and decayed data.  
At each time step $i$ the algorithm receives a bit $x_i$; at each time
step $j$, the algorithm is required to report an approximation
$\hat{F}(x_1, \ldots, x_{j})$ to a function $F(x_1, \ldots, x_{j})$
and be $\eps$-differentially private over the entire data seen thus far. We use the notion of $(\delta,
\gamma)$-utility, satisfied by algorithms that at any time step $j$
output a value $\hat{F}(x_1, \ldots, x_j)$ which is within $\delta$
absolute error from $F(x_1, \ldots, x_j)$ with probability
$1-\gamma$. Below we summarize our results for sufficiently small $\gamma$ (results for larger $\gamma$ can be found in the body of the paper):
\begin{itemize}
  \setlength{\itemsep}{0.5pt}
  \setlength{\parskip}{0pt}
  \setlength{\parsep}{0pt}
\item ({\em Window Sum})
  The \emph{window sum} problem with window
  size $W$ requires estimating $F_w(j, W) = \sum_{i = j - W +
    1}^j{x_i}$ for each $j$.  Further, the whole sequence $F_w$ of
  outputs, for all $j$, should be $\eps$-differentially private.
  
  \smallskip
We present an algorithm that achieves $(\delta, \gamma)$-utility for
$\delta = O(\frac{1}{\eps}\log W \log \frac{1}{\gamma})$ (in the regime
$\log W \geq \log \frac{1}{\gamma}$ ). 
%We also present an algorithm that can approximate
%window sum simultaneously for all window sizes $W$ and for each
%particular $W$ achieves error comparable to the specialized
%algorithm. Note that 
While a window sum can be reduced to computing the
difference of two running sums, existing running sum
algorithms~\cite{dwork-continual,chan2010private} achieve error
$\delta = \Theta(\frac{1}{\eps}\log T\log \frac{1}{\gamma})$, which can be much
larger than the range $W$ of $F_w$, and therefore, as bad as the
trivial algorithm that outputs a fixed value independently of the input.

\smallskip
We also present a lower bound of $\Omega(\min\{W/2, \frac{1}{\eps}\log
\frac{1}{\gamma}\})$. Note that the dependence on the error probability
$\gamma$ is optimal. The $W/2$ term in the lower bound is unavoidable,
as the trivial algorithm which outputs $W/2$ at every time step
achieves additive approximation $W/2$ and is perfectly private. This
lower bound generalizes a previous lower for the running sum
problem~\cite{dwork-continual}. 

\smallskip
\item ({\em Exponential Decay}) The \emph{exponential decay sum}
  problem is to estimate $ F_e(j, \alpha) = \sum_{i = 1}^j{x_i
    \alpha^{j-i}}$ accurately, while the whole sequence $F_e$ of
  outputs, for all $j$, should be $\eps$-differentially private.

\smallskip
  We present an algorithm that achieves $(\delta, \gamma)$-utility
  with   $\delta = O(\frac{1}{\eps}\log \frac{\alpha}{1-\alpha} \log
  \frac{1}{\gamma})$.  We also  present a lower bound of
  $\Omega\left(\min\left\{\frac{\alpha}{1-\alpha}, \frac{\log
        (1/\gamma)}{\eps}\right\}\right)$. Once again, the dependence
  on the error probability $\gamma$ is optimal. Unlike $F_w$, $F_e$ at each time step
  depends on the entire sequence of updates; nevertheless, our
  algorithm achieves bounded error, polylogarithmic in the range of $F_e$.

\smallskip
\item ({\em Polynomial Decay}) The \emph{polynomial decay sum} problem
  is to estimate $ F_p(j, c) = \sum_{i = 1}^j{\frac{x_i}{(j-i+1)^c}}$
  accurately, while the whole sequence $F_p$ of outputs, for all $j$,
  should be $\eps$-differentially private.

\smallskip
  We present an algorithm that for each $j$ returns $(1\pm\beta) F_p(j,c) \pm
  \left(\frac{1}{c\beta^2}\log\frac{1}{1-\beta}\right)\log
  \frac{1}{\gamma}$ with probability $1-\gamma$. We also present a lower
  bound of $\Omega\left(1 - \frac{\eps^{c-1}}{\log^{c-1}
      (1/\gamma)}\right)$ against {\em purely} additive error. Polynomial decay presents a greater challenge than window sums or exponential decay since there is no direct way to combine a polynomial decay sum over an
interval $[a, b]$ and 
%a polynomial decay sum over another interval
$[b, c]$ into a polynomial decay sum over $[a, c]$. We develop a general technique that works on a large class of decay sum functions (including polynomial decay) and reduces the problem of estimating the decay sum to keeping multiple window sums in parallel. The technique results in a bi-criteria approximation, because of which our lower and upper bounds are incomparable for this problem. 
      
\end{itemize}

In comparison with the simple randomized response strategy~\cite{rr} (i.e.~with probability $1/2 - \eps/2$ change update $x_i$ to $1-x_i$ and keep exact statistics of the changed input), our algorithms achieve exponentially smaller additive error: randomized response leads to estimators with standard deviation proportional to the energy of the decay function, while our estimators have standard deviation polylogarithmic in the energy. Technically, 
\begin{packed_item}
\item
Our algorithms keep dyadic tree data structures as is natural and also used in~\cite{dwork-continual,chan2010private} and elsewhere. However, in order to provide estimates with error polylogarithmic in the range of the decay function, we need to treat the dyadic tree data structure in non-uniform manner: either adding different noise at different nodes, or weighing the contribution of an update to different nodes differently, which is our technical contribution.
\item
We derive all our lower bounds from a common framework, that is inspired by work on differentially private combinatorial optimization. This extends prior work in two ways:  they apply to decay sum
problems that have not been considered before, and they apply against the
weaker $(\delta, \gamma)$-utility  guarantee  (rather than requiring
that all queries are accurate, as in~\cite{dwork-continual}. 
\end{packed_item}

%We derive all our lower bounds from a common framework, that is inspired by work on %differentially private combinatorial optimization. 
%The lower bounds proven in~\cite{dwork-continual} themselves fit into this framework. 

%Technically, one keeps a dyadic tree. 
%different tree nodes non-uniformly: either weighing updates to different nodes differently and/or adding different noise at different nodes.  

\junk{
Differential privacy under continual updatesof data stream is an exciting research
direction, and the study of other decayed functions will be of great
interest. In particular, user-level privacy and sketches over windows
are  interesting problems for future study.
}

\noindent
\textbf{Detailed discussion of prior work.} The problem of tracking
statistics on dynamic data while preserving privacy under continual
observation is introduced in~\cite{dwork-continual} 
%by Dwork, Pitassi, Naor, and Rothblum
 (a preliminary version was presented in an invited
talk by Dwork~\cite{cynthia}).
In~\cite{dwork-continual}, an algorithm private under continual
observation is presented for the running sum problem. For any fixed
time step, their algorithm achieves additive error of
$O(\frac{1}{\eps}\log^{1.5} T)$ with constant probability, where $T$
is an upper bound on the maximum size of the input, 
%and is assumed to be 
known to the algorithm. \junk{This algorithm, similarly to our
work, also uses a dyadic tree datastructure but in a slightly
different way: the dyadic tree is used only to store noise variables
rather than noisy counters.} 
Independently, ~\cite{chan2010private} presented a continually
private algorithm for the running sum problem that at any step $j$, 
guarantees an additive error of $O(\frac{1}{\eps}\log^{1.5} j)$
with constant probability. This matches~\cite{dwork-continual}, while not knowing $T$. 
%and also uses a dyadic tree data
%structure. However,  the error grows with $T$ since 
% does grow as
%the size of the processed input grows: 
%at any step $j$, they
%guarantee an additive error of $O(\frac{1}{\eps}\log^{1.5} j)$
%with constant probability.
%, matching the upper bound for the algorithm
%in~\cite{dwork-continual} without the need to specify an explicit
%bound $T$. 

The algorithm of ~\cite{chan2010private} is  related
to our work: our algorithm for window sum reduces
to their algorithm for running sum when the size of the window
coincides with the size of the input. However, if their algorithm or
the algorithm in~\cite{dwork-continual,song-full} is used directly to compute
window sums, then the error at time $j$ will be on the order of
$\log^{1.5}j$ and for large $j$ will overcome the window size.
Further, algorithms in~\cite{chan2010private,dwork-continual,song-full} do not work for
decayed sums.  

\junk{The
same holds for the algorithm for range queries given in the full
version of the Chan, Shi, and Song paper~\cite{song-full}.  There
is no known continually private algorithm for the running sum problem
that outputs an estimate at all time steps and has bounded error at
any fixed time step $j$ with constant probability, where the bound on
the error is independent of $j$. The basic randomized response
technique, on the other hand, does provide error guarantees for
estimating the window sum (or other decay sum functions) at time $j$
that are independent of $j$. However, these guarantees are only on the
order of the square root of the range of the window sum. By contrast,
our work shows that for decay sum problems, continually private
algorithms can be designed so that the error at any time step $j$ is
bounded, polylogarithmic in the range of the decay sum function, and
independent of $j$. Moreover, we show how to compute private estimates
simultaneously for all window sizes while the estimate for any
particular window size $W$ is accurate to within a polylogarithmic
factor of \emph{$W$}. This is a major difference from the range query
algorithm in~\cite{song-full}, which gives the same accuracy
guarantees for very small and very large windows. In order to achieve
nonuniform guarantees for different window sizes, as well as to design
accurate continually private algorithms for exponential and polynomial
decay sums,}

\cite{dwork-continual} shows how to transform a  private {\em streaming} algorithm that satisfies a monotonicity
property to a  private, continual algorithm.
% under continual observation.
 \junk{(the transformation works in the stronger user-level
privacy model as long as the base algorithm is private in the
user-level model). The monotonicity property required for their
transformation is that when the base algorithm is run on prefixes of
the input of different lengths, with high probability the output of
the algorithm does not change by a large additive quantity too many
times.} 
However, estimating decayed sums does not have the monotonicity property. 

\section{Notation and Preliminaries}

\noindent
{\bf Online Data Model}
We consider online problems with binary input: at each time step $i$
the algorithm receives input $x_i \in [a, b]$; and is required to report an
approximation $\hat{F}(x_1, \ldots, x_{i})$ to a function $F(x_1,
\ldots, x_{i})$. We present oue upper bounds for $a = 0, b =
1$. For general $a, b$,  our absolute error bounds scale linearly in
$b - a$.

\medskip
\noindent
{\bf Decayed Sum Problems}
The functions $F$ we are interested in approximating are \emph{decayed
  sum} functions. Consider a non-increasing function $g: \mathbb{N}
\rightarrow \mathbb{R}^+$ such that $g(0) = 1$. The \emph{decayed sum
  induced by $g$} is the function
%\begin{equation}
%  \label{eq:decay-defn}
  $F(j) = F(x_1, \ldots, x_j) = \sum_{i = 1}^j{x_i g(j - i)}$.
%\end{equation}
I.e., $F$ is the convolution of the input
% $x_1, x_2, \ldots$ 
and a
non-increasing function $g$. The decayed sum problems we 
consider are defined below
\junk{$F_w(j, W)$, $F_e(j, \alpha)$, and $F_p(j, c)$ as defined
  earlier.}

\begin{itemize}
\item when $g(i) = 1 \forall i$, the \emph{running sum} problem
  (considered in \cite{chan2010private,cynthia,dwork-continual}):
  %\begin{equation*}
    $F_s(j) = \sum_{i = 1}^j{x_i}$.
  %\end{equation*}

\item when $g(i) = \mathbf{1}_{\{i < W\}}$, the \emph{window sum} problem
  (with window size $W$):
  %\begin{equation*}
    $F_w(j, W) = \sum_{i = j - W + 1}^j{x_i}$.
  %\end{equation*}
  To simplify notation, in the above definition we assume that $x_i =
  0$ for all $i \leq 0$.

\item when $g(i) = \alpha^i$ ($\alpha < 1$), the \emph{exponential
    decay sum} problem:
  %\begin{equation*}
  $F_e(j, \alpha) = \sum_{i = 1}^j{x_i \alpha^{j-i}}$.
  %\end{equation*}

\item when $g(i) = (i+1)^{-c}$ ($c > 1$), the \emph{polynomial decay sum}
  problem: 
  %\begin{equation*}
    $F_p(j, c) = \sum_{i = 1}^j{\frac{x_i}{(j-i+1)^c}}$.
  %\end{equation*}
\end{itemize}
The last three problems have not been considered in the differential privacy
literature before, and specifically not in the continual observation
model. The problems of keeping event counts and other statistics over
windows~\cite{datar2002maintaining} and keeping decayed (in particular
exponential and polynomial decay) sums~\cite{cohen2003maintaining} have
been studied in the field of small space streaming algorithms.

\medskip
\noindent
{\bf Differential Privacy}
We use the standard definition of differential privacy, applied to the
online data model:
\begin{definition}[\cite{DMNS,dwork-continual}]
  Let $\Alg$ be a randomized online algorithm that at time step $j$
  outputs $\hat{F}(x_1, \ldots, x_j) \in \mathbb{R}$. $\Alg$ satisfies
  \emph{$\eps$-differential privacy} if for all $T \in \mathbb{Z}$,
  for all measurable subsets $S  \subseteq \mathbb{R}^T$,  all
  possible inputs $x_1, \ldots, x_T$, all $j$ and all $x'_j$ (where probability is over the coin throws of $\Alg$)
  \begin{equation*}
    \Pr[(\hat{F}(x_1, \ldots,x_j,\ldots,  x_k))_{k = 1}^T \in S] \leq e^\eps
    \Pr[(\hat{F}(x_1, \ldots, x'_j,\ldots,  x_k))_{k = 1}^T \in S].
  \end{equation*}
\end{definition}
This is the basic definition of differential privacy as
in~\cite{DMNS}, but with the modification that the algorithm receives
the input online and produces output at every step, and the whole sequence
of outputs is available to an adversary. This model of privacy for
online algorithms operating on time series data, termed
\emph{privacy under continual observation}, was introduced
by~\cite{cynthia,dwork-continual}.

We use the following basic facts about differential privacy. The
first theorem gives a simple way to achieve differential privacy for
algorithms with numerical output, based on adding random noise scaled
according to the sensitivity of the statistic being computed. The second fact is that composing multiple privacy mechanisms results
in smooth privacy loss.
\begin{theorem}[\cite{DMNS}]
  \label{thm:laplace}
  For a function $F: [a, b]^T \rightarrow \mathbb{R}^d$, let the
  \emph{sensitivity} of $F$, $S_F$ be the smallest real number that
  satisfies $  \forall x_1, \ldots, x_T, \forall j \in [T], \forall x'_j \in [a,
  b]:$
\begin{align*}
%  \forall x_1, \ldots, x_T, \forall j \in [T]&, \forall x'_j \in [a,
%  b]:\\
  &\|F(x_1, \ldots,x_j,\ldots,  x_T) - F(x_1, \ldots, x'_j,\ldots,  x_T)\|_1
  \leq S_F
\end{align*}
Then an algorithm that on input $x_1, \ldots, x_T$ outputs
$\hat{F}(x_1, \ldots, x_T) = F(x_1, \ldots, x_T) + \Lap(S_F/\eps)^d$
satisfies $\eps$-differential privacy, where $\Lap(\lambda)^d$ is a
sample of $d$ independent Laplace random variables with mean 0 and
scale parameter $\lambda$.
\end{theorem}

\begin{theorem}[\cite{DMNS}]
  \label{thm:composition}
  Let algorithm $\Alg_1$ satisfy $\eps_1$-differential privacy and
  algorithm $\Alg_2$ satisfy $\eps_2$-differential privacy. Then an
  algorithm $\Alg$ that on input $\vec{x} = \{x_1, \ldots, x_T\}$
  outputs $\Alg(\Alg_1(\vec{x}), \Alg_2(\vec{x}))$ satisfies $(\eps_1
  + \eps_2)$-differential privacy.
\end{theorem}

\medskip
\noindent
{\bf Utility}
We adopt the following, commonly used notion of utility:
%, commonly used in
%differential privacy research.

\begin{definition}
  Let $\Alg$ be a randomized online algorithm that at time step $j$
  outputs $\hat{F}(x_1, \ldots, x_j) \in \mathbb{R}$. Then, $\Alg$
  achieves $(\delta, \gamma)$-utility with respect to a function $F$, if
  for all $j$,
  %\begin{equation*}
    $\Pr[|\hat{F}(x_1, \ldots, x_j) - F(x_1, \ldots, x_j)| > \delta] < \gamma$.
  %\end{equation*}
\end{definition}

\begin{figure}[t!]
\begin{center}
\psfrag{T}[cc][cc]{{\tiny $\tree=\tree(L, U)$}}
\psfrag{T'}[cc][cc]{{\tiny $\blue{\tree'=\tree(L+4, U)}$}}
\psfrag{[L,U]}[cc][cc]{{\tiny $[L,U]$}}
\psfrag{[L,L+3]}[cc][cc]{{\tiny $[L,L\!+\!3]$}}
\psfrag{[L+4,U]}[cc][cc]{{\tiny $[L\!+\!4,U]$}}
\psfrag{[L,L+1]}[cc][cc]{{\tiny $[L,L\!+\!1]$}}
\psfrag{[L+2,L+3]}[cc][cc]{{\tiny $[ L\!\!+\!\!2 \! , \!L\!\!+\!\!3 ]$}}
\psfrag{[L+4,L+5]}[cc][cc]{{\tiny $[L \!\!+\!\!4 \! , \!L\!\!+\!\!5]$}}
\psfrag{[L+6,U]}[cc][cc]{{\tiny $[L\!+\!6,U]$}}
\psfrag{L}[cc][cc]{{\tiny $L$}}
\psfrag{L+1}[cc][cc]{{\tiny $L+1$}}
\psfrag{L+2}[cc][cc]{{\tiny $L+2$}}
\psfrag{L+3}[cc][cc]{{\tiny $L+3$}}
\psfrag{L+4}[cc][cc]{{\tiny $L+4$}}
\psfrag{L+5}[cc][cc]{{\tiny $u\!=\!L\!+\!5$}}
\psfrag{L+6}[cc][cc]{{\tiny $L+6$}}
\psfrag{U}[cc][cc]{{\tiny $U$}}
\psfrag{c_{LL+3}}[cc][cc]{{\tiny $c_{Lu'}=c_{L, L+3}$}}
\psfrag{c_{L+4L+5}}[cc][cc]{{\tiny $s(u,\tree')=c_{L+4, L+5}$}}
\psfrag{s(u,T)}[cc][cc]{{\tiny $s(u, \tree)=c_{L, L+3}+c_{L+4, L+5}$}}
\includegraphics[width=0.95\columnwidth]{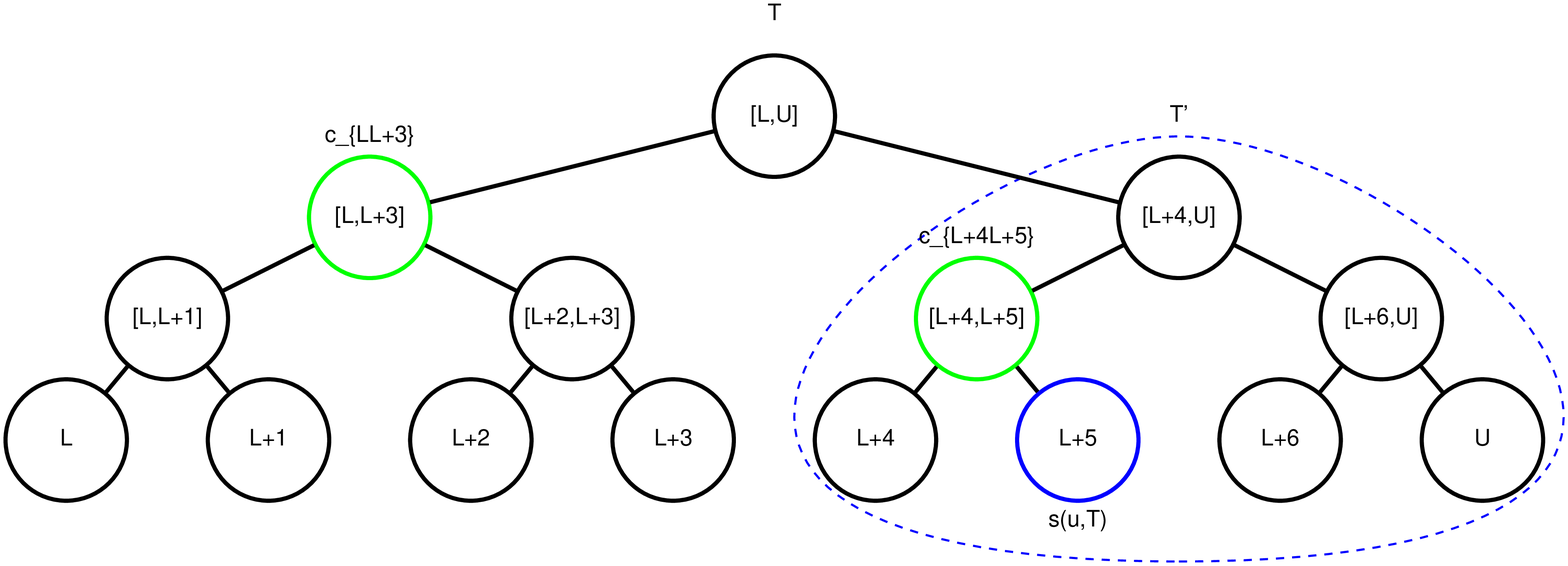}
\caption{Dyadic tree data structure. In this example, $u=L+5$ is shown in a blue node, $u'=L+3$, and the prefix sum $s(u, \tree )=c_{L, L+3}+c_{L+4, L+5}$ is obtained by adding the counters at the two green nodes $[L,L+3]$ and $[L+4,L+5]$.}
\label{fig:DyadicTree}
\end{center}
\end{figure}

\medskip
\noindent
{\bf Dyadic Tree Datastructure}
We repeatedly use the following \emph{dyadic tree} data
structure which is common in algorithmics. This data structure is a balanced augmented search tree and
variants of it are common in much algorithmic work.
%(e.g.~a similar data structure was used to approximate range queries
%and heavy hitters in small space~\cite{cormode2005improved}).

Let $\tree(L, U)$ be a complete binary tree of height $h =
\log(U-L+1) + 1$ (assuming, for simplicity, that $U-L + 1$ is a power
of 2). The leaves of the tree are indexed by the integers $L, L+1,
\ldots, U$, and if two sibling nodes are indexed by the intervals
$[l_1, u_1]$ and $[l_1 = u_1 + 1, u_2]$, then their parent is indexed
by $[l_1, u_2]$. Note that at level $k$ of the tree (the leaves being
at level $1$), the indexing intervals have the form $[L +
(i-1)2^{k-1}, L + i2^{k-1} -1]$ for $i \in [1, 2^{h-k}]$. We call a
node whose indexing interval precedes its sibling's indexing interval
a \emph{left node}; the sibling of a \emph{left node} is a \emph{right
  node}. With each node we associate a variable: for the node indexed
by $[l, u]$, the associated variable is denoted $c_{lu}$. Given a tree
$\tree = \tree(L, U)$ and a prefix interval $[L, u]$,  we define
%we define the
function $s(u, \tree)$ recursively:
\begin{packed_item}
\item if $[L, u]$ indexes a node in $\tree$, then $s(u, \tree) = c_{Lu}$;
\item otherwise, let $u'$ be the largest integer less than $u$ such
  that $[L, u']$ indexes some node in $\tree$; equivalently, $u'$ is
  the largest integer less than $u$ that can be written as $u' = L +
  2^{k-1} - 1$. Let $\tree'$ be the subtree of $\tree$ rooted at the
  sibling of $[L, u']$ (indexed by $[u' + 1, u' + 2^{k-1}]$); then
  $s(u, \tree) = c_{Lu'} + s(u, \tree')$.
\end{packed_item}

The following lemma is essential to our analysis and can be easily
proved by induction.
%\begin{claim}
%  \label{cl:log-terms}
%  There exist $r \leq \lceil \log(u - L + 1) \rceil$ integers $L = u_0, u_1, \ldots
%  u_{r}, u_{r+1} = u$ such that $s(u, \tree) = c_{u_0u_1} + \sum_{k =
%    1}^{r}{c_{u_k+1,u_{k+1}}}$.  Furthermore, all nodes indexed by
%  $[u_k+1, u_{k+1}]$ are left nodes in $\tree$, and each node is in a
%  different level of $\tree$.
%\end{claim}
%
\begin{lemma}
  \label{lm:log-terms}
  There exist $r \leq \lceil \log(u - L + 1) \rceil$ integers $u_1, \ldots
  u_{r}= u$ such that $s(u, \tree) = c_{L u_1} + \sum_{k =
    1}^{r-1}{c_{u_k+1,u_{k+1}}}$.  Furthermore, all nodes indexed by
  $[u_k+1, u_{k+1}]$ are left nodes in $\tree$, and each node is in a
  different level of $\tree$.
\end{lemma}
\begin{proof}
  The integers $u_1, \ldots, u_r$ are given directly by the recursive
definition of $s(u, \tree)$. To bound $r$, consider that at each step
in the recursion, unless $[L, u]$ indexes a node in $\tree$, the tree
$\tree'$ has at most half the number of leaves of the smallest subtree
of $\tree$ that contains $u$ as a leaf. Initially the smallest subtree
that contains $u$ as a leaf has number of leaves equal to the smallest
power of 2 greater than or equal to $u-L+1$, i.e. the number of leaves
initially is $2^{\lceil\log(u-L+1)\rceil}$. The recursion stops when
we reach a tree with only a single node, and, therefore, we make at
most $\lceil\log(u-L+1)\rceil$ recursive calls. The bound on $r$
follows.

The condition that all nodes are left siblings follows from the fact
each node is indexed by an interval that contains the leftmost leaf of
the current subtree.

Finally, notice that the only way to pick two nodes on the same level
is if after picking $u'$, in the next step of the recursion we pick
the root of $\tree'$. However, in this case we would have picked the
parent of $[L, u']$ instead of $[L, u']$, a contradiction.  
\end{proof}

\medskip
\noindent
{\bf Chernoff Bound for Laplace Variables}
We will use the following Chernoff bound for sums of independent
Laplace random variables. 
\begin{lemma}
  \label{lm:chernoff}
  Let $s_1,\ldots, s_n$ be independent Laplace random variables such
  that $s_i \sim \Lap(b_i)$. Denote $S = \sum_{i = 1}^n{s_i}$ and
  $\sigma = \sqrt{2\sum_{i = 1}^n{b_i^2}}$. Then, for all $\lambda <
  \min_i \frac{0.75}{b_i}$, we have $\Pr[|S| \geq t\sigma] \leq 2\exp(0.75
  \lambda^2 \sigma^2 - \lambda t\sigma).$
\end{lemma}
\begin{proof}
  We use the standard technique of bounding the moment generating
function of $S$ and applying Markov's inequality. Details follow.

Since the distribution of $S$ is symmetric, we have $\Pr[|S| \geq
t\sigma] = 2\Pr[S \geq t\sigma]$. For any $\lambda$, we have:
\begin{align}
  \Pr[S \geq t\sigma] &= \Pr[e^{\lambda S} \geq e^{\lambda t
    \sigma}]\notag\\
  &\leq \frac{\E[e^{\lambda S}]}{e^{\lambda t\sigma}} = \frac{\prod_{i
      = 1}^n{\E[e^{\lambda s_i}}]}{e^{\lambda t \sigma}}\label{eq:chernoff-trick}
\end{align}
For $\lambda < 1/b_i$, the moment generating function of the
Laplace random variable $s_i$ is $\E[e^{\lambda s_i}] =
1/(1-\lambda^2b_i^2)$. Assuming $\lambda b_i \leq .75$, we have
\begin{equation*}
  \E[e^{\lambda s_i}] = \frac{1}{1 - \lambda^2b_i^2} <
  \exp(-\frac{3}{2}\lambda^2 b_i^2). 
\end{equation*}
Substituting into (\ref{eq:chernoff-trick}), we get
\begin{equation*}
  \Pr[S \geq t\sigma] \leq \exp(-\frac{3}{2}\lambda^2\sum_{i =
    1}^n{b_i^2} - \lambda t \sigma),
\end{equation*}
as desired.
\end{proof}

\section{Upper Bounds}

\subsection{Window Sum}

\begin{figure}[t!]
\begin{center}
\psfrag{W}[cc][cc]{{\tiny $W=4$}}
\psfrag{T1}[cc][cc]{{\tiny $\tree_1=\tree(1,W)=\tree(1, 4)$}}
\psfrag{T2}[cc][cc]{{\tiny $\tree_2=\tree(W+1,2W)=\tree(5, 8)$}}
\psfrag{[L,L+3]}[cc][cc]{{\tiny $[1,4]$}}
\psfrag{[L+4,U]}[cc][cc]{{\tiny $[5,8]$}}
\psfrag{[L,L+1]}[cc][cc]{{\tiny $[1,2]$}}
\psfrag{[L+2,L+3]}[cc][cc]{{\tiny $[3,4]$}}
\psfrag{[L+4,L+5]}[cc][cc]{{\tiny $[5,6]$}}
\psfrag{[L+6,U]}[cc][cc]{{\tiny $[7,8]$}}
\psfrag{L}[cc][cc]{{\tiny $1$}}
\psfrag{L+1}[cc][cc]{{\tiny $2$}}
\psfrag{L+2}[cc][cc]{{\tiny $3$}}
\psfrag{L+3}[cc][cc]{{\tiny $4$}}
\psfrag{L+4}[cc][cc]{{\tiny $5$}}
\psfrag{L+5}[cc][cc]{{\tiny $6$}}
\psfrag{L+6}[cc][cc]{{\tiny $7$}}
\psfrag{U}[cc][cc]{{\tiny $8$}}
\psfrag{c_{L}}[cc][cc]{{\tiny $c_{1}\!=\!x_1\!+\!z_1$}}
\psfrag{c_{L+1}}[cc][cc]{{\tiny $c_{2}\!=\!x_2\!+\!z_2$}}
\psfrag{c_{L+2}}[cc][cc]{{\tiny $c_{3}\!=\!x_3\!+\!z_3$}}
\psfrag{c_{L+3}}[cc][cc]{{\tiny $c_{4}\!=\!x_4\!+\!z_4$}}
\psfrag{c_{L+4}}[cc][cc]{{\tiny $c_{5}\!=\!x_5\!+\!z_5$}}
\psfrag{c_{L+5}}[cc][cc]{{\tiny $c_{6}\!=\!x_6\!+\!z_6$}}
\psfrag{c_{L+6}}[cc][cc]{{\tiny $c_{7}\!=\!x_7\!+\!z_7$}}
\psfrag{c_{U}}[cc][cc]{{\tiny $c_{8}\!=\!z_8$}}
\psfrag{c_{LL+1}}[cc][cc]{{\tiny $c_{12}=x_1+x_2+z_{12}$}}
\psfrag{c_{L+2L+3}}[cc][cc]{{\tiny $c_{34}=x_3+x_4+z_{34}$}}
\psfrag{c_{L+4L+5}}[cc][cc]{{\tiny $c_{56}=x_5+x_6+z_{56}$}}
\psfrag{c_{L+6U}}[cc][cc]{{\tiny $c_{78}=x_7+z_{78}$}}
\psfrag{c_{LL+3}}[cc][cc]{{\tiny $c_{14}=x_1+x_2+x_3+x_4+z_{14}$}}
\psfrag{c_{L+4U}}[cc][cc]{{\tiny $c_{58}=x_5+x_6+x_7+z_{58}$}}
\psfrag{s(u,T)}[cc][cc]{{\tiny $\hat{F}_w(7, W=4) = \red{c_{14}}-(\green{c_{12}+c_3})+ (\blue{c_{56}+c_{7}})$}}
%\psfrag{s(u,T)}[cc][cc]{{\tiny $\hat{F}_w(7, W=4) = s(4, \tree_{1}) - s(3,
%      \tree_{1})+  s(7, \tree_2)=c_{14}-(c_{12}+c_3)+ (c_{56}+c_{7})$}}
\includegraphics[width=0.95\columnwidth]{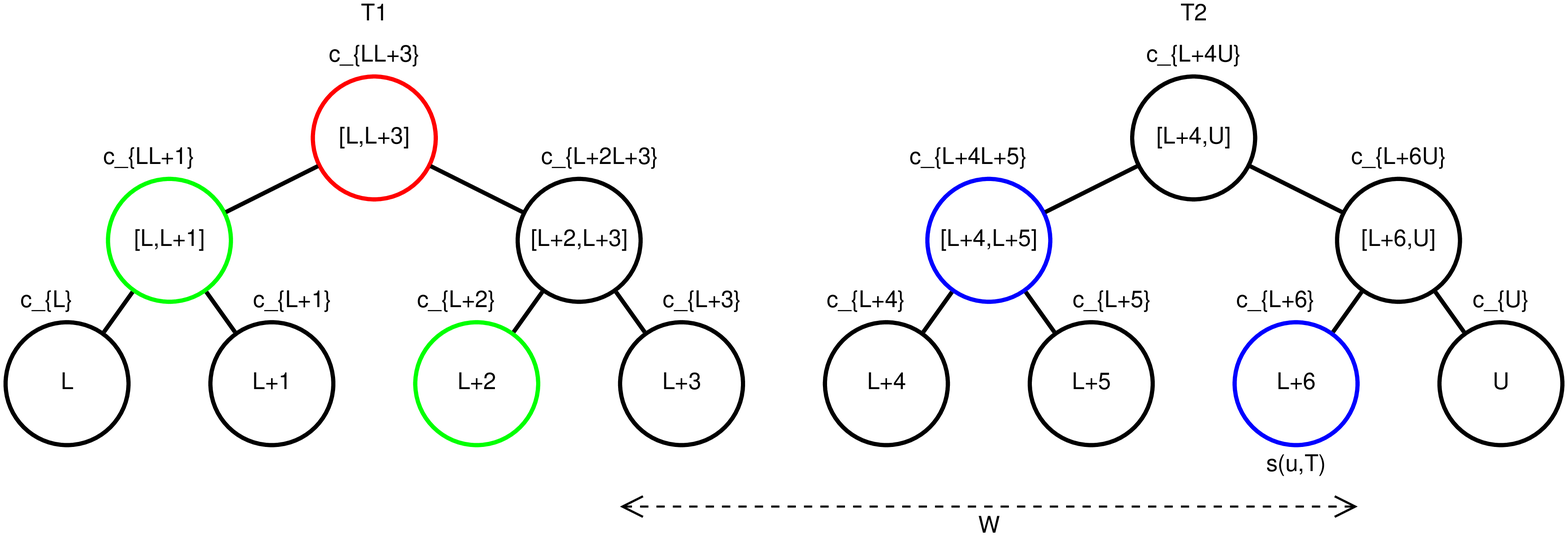}
\caption{Window sum for a window size $W=4$. This example illustrates the algorithm at time $i=7$, and the output is $\hat{F}_w(7, W=4) = s(W, \tree_{1}) - s(3, \tree_{1})+  s(7, \tree_2)= c_{14}-(c_{12}+c_3)+ (c_{56}+c_{7})=x_4+x_5+x_6+x_7+z_{14}-z_{12}-z_3+z_{56}+z_{7}$, where $z_{lu}$ denotes the noise at node $[l,u]$.}
\label{fig:DyadicTree-WindowSum}
\end{center}
\end{figure}

A key observation for computing window sums with error
polylogarithmically bounded in $W$ is that, unlike with running sum,
only the lowest $\log W + 1$ layers of the dyadic tree are necessary
to compute window sum. However, if we keep a dyadic tree for every
window of size $W$, each update will contribute to more than $W$
variables, resulting in data structures with large sensitivity, which,
for differential privacy, translates into more noise. Our main idea is
that instead of keeping a dyadic tree for every window, we can divide
the input into blocks of size $W$, and view the windows that span two
blocks as the union of a suffix and a prefix of two blocks. 
%This only
%requires a constant factor more counters to estimate a single window
%sum.

The algorithm \textsc{WindowSum} is shown as
Algorithm~\ref{alg:winsum}. In the remainder of this section we assume
that $W$ is an exact power of 2.

\begin{algorithm}[t]
  \caption{\textsc{WindowsSum}}\label{alg:winsum}
  {\fontsize{9}{9}\selectfont
  \begin{algorithmic}
    \STATE For $k \geq 1$, define $\tree_k = \tree((k-1)W + 1, kW)$,
    with all $c_{lu}$ initialized to $\Lap((\log W+1)/\eps)$.

    \FORALL{inputs $x_i$}
    \STATE add $x_i$ to all $c_{lu}$ in $\tree_{\lceil i/W
      \rceil}$ such that $i \in [l, u]$.
    \STATE output: $\hat{F}_w(i, W) = s((k-1)W, \tree_{k-1}) - s(i - W, \tree_{k-1})+  s(i, \tree_k)$, where $k = \lceil i/W \rceil$.
    \ENDFOR
  \end{algorithmic}
}
\end{algorithm}

\begin{theorem}
  \label{thm:winsum}
  %There exists a constant $K$ s.t.~
  \textsc{WindowSum} satisfies
  $\eps$-differential privacy, and achieves $(\delta, \gamma)$-utility
  with
  \begin{equation*}
    \delta = \begin{cases}
      O(\frac{1}{\eps}\log^{1.5} W \log^{0.5} \frac{1}{\gamma}), &\log
      W  \geq  \log\frac{1}{\gamma}\\
      O(\frac{1}{\eps}\log W \log\frac{1}{\gamma}), &\log W  <  \log\frac{1}{\gamma}
      \end{cases}
  \end{equation*}
  Furthermore, \textsc{WindowSum} can be implemented to use $O(W)$
  words of space and to run in $O(\log W)$ time per update.
\end{theorem}
\begin{proof}
  \textbf{Privacy.} Observe that any variable $c_{lu}$ used to compute
  $\hat{F}_w(j, W)$ satisfies $l \leq u \leq j$. Therefore, the
  counters $c_{lu}$ that contribute to $\hat{F}_w(j, W)$ will not be
  updated after time step $j$ and $\hat{F}_w(j, W)$ will be
  identically distributed if it is computed at any time step $T \geq
  j$, so for the analysis we can assume that all outputs are
  produced at time step $T$. Next we fix $T$ and argue that
  \textsc{WindowSum} is $\eps$-differentially private for inputs of
  size $T$. Since the choice of $T$ is arbitrary, privacy for all $T$
  follows. For this purpose, let $\vec{c}(\vec{x})$ be the vector of
  the values of all variables (in an arbitrary order) $c_{lu}$ such
  that $u \leq T$ when the input is $\vec{x} = (x_1, \ldots,
  x_T)$. Let also $\vec{c}_0(\vec{x})$ be $\vec{c}(\vec{x})$ with the
  initializing Laplace noise removed. 
  %Observe that, 
  Since each $x_j$
  contributes to exactly $\log W+1$ variables $c_{lu}$
  \begin{equation*}
  % \label{eq:win-sens}
    \forall j \in [T] \mbox{, } \forall x'_j \in [0, 1]: \|\vec{c}_0(x_1, \ldots, x_j, \ldots, x_T) -
    \vec{c}_0(x_1, \ldots, x'_j, \ldots, x_T)\|_1 \leq \log W + 1.
  \end{equation*}
  Differential privacy of $\vec{c}(\vec{x})$ follows from above
  % (\ref{eq:win-sens}) 
  and
  Theorem~\ref{thm:laplace}. Since the sequence of outputs of \textsc{WindowSum}
  up to time step $T$ is a deterministic function of
  $\vec{c}(\vec{x})$, privacy of \textsc{WindowSum} follows.

  \textbf{Accuracy.} It is easy to see that $\E \hat{F}_w(j, W) =
  F_w(j, W)$. By Lemma~\ref{lm:log-terms}, for each $k$ and each $u$,
  $s( u, \tree_k)$ is the sum of at most $\log W$ random variables,
  each with variance $2(\log W+1)^2/\eps^2$. Therefore, the
  standard deviation $\sigma$ of $\hat{F}_w(j, W)$ is $O(\log^{1.5}
  W/\epsilon)$.

  We consider two cases. For the first case, let $t = 2\sqrt{\ln
    (1/\gamma)}$ and $\lambda = \phi\frac{t}{\sigma}$ for a constant
  $\phi$ to be determined later. By Lemma~\ref{lm:chernoff}, as long
  as $\lambda < 0.75\eps/(\log W + 1)$, we have that
  $\Pr[|\hat{F}_w(j, W) - F_w(j,W)| > C\sqrt{\log (1/\gamma)}\sigma] <
  \gamma$ for some fixed constant $C$ that depends on $\phi$. A
  calculation shows that as long as $\log W \geq \log (1/\gamma)$,
  the minimum value of $\phi$ such that the
  constraint on $\lambda$ holds can be bounded below by a
  constant. This completes the analysis of the first case.
  
  For the second case, when $\log W< \log (1/\gamma)$, we set the
  following parameters: $\eta = \log_{\ln(1/\gamma)}{\ln W}$ (notice
  that $\eta < 1$); $t = C'\frac{\ln (1/\gamma)}{\sqrt{\ln W}}$, and
  $\lambda = \frac{t^{\eta/(2-\eta)}}{\sigma} =
  \frac{{C'}^{\eta/(2-\eta)}\sqrt{\ln W}}{\sigma}$, where $C'$ is a
  constant chosen so that $\lambda < 0.75\eps/(\log W + 1)$
  holds. Applying Lemma~\ref{lm:chernoff}, we have that for a value
  $C$ that depends on $C'$, $\Pr[S \geq C \frac{1}{\eps}
  \log(1/\gamma)\log W] \leq \exp(\Omega(t^{2/(2-\eta)})) =
  \exp(-\Omega(\ln 1/\gamma))$.
  
  For the running time and space complexity analysis, notice that each
  update requires accessing $O(\log W)$ nodes, and that only the last
  two dyadic trees need to be stored. 
  %\qed
\end{proof}

We can also show that we can approximate window sums
\emph{simultaneously} for all window sizes and preserve privacy under
continual observation.  Our approximation  is different
  for different window sizes $W$, and for any particular $W$, it is
almost the same as that of Theorem~\ref{thm:winsum}. Details can be
found in Appendix~\ref{app:all-winsum}. 

\section{Window Sum Simultaneously for all $W$}
\label{app:all-winsum}

Here we give an algorithm that works simultaneously for all window
sizes.  Our main observation \junk{for this algorithm} is that if for
window size $W$ we divide the input into blocks of size $W' \in [W,
2W]$ instead of exactly $W$ as in \textsc{WindowSum}, then we can
store all necessary dyadic tree datastructures as subtrees of a single
dyadic tree. However, storing the whole dyadic tree with the same
noise at any level will result in error of size $\Omega(\log^{1.5}T)$
for all $W$. Instead, we want to make sure that within a subtree of
height $h$, the noise added to any variable is proportional to $h$. To
achieve this, we use a different privacy parameter $\eps_k$ at level
$k$ of the dyadic tree and ensure that the sum of privacy parameters
converges to $\eps$.

Let $\beta>1$ be a parameter and $\zeta(\cdot)$ be the Riemann zeta
function: $\zeta(\beta) = \sum_{1}^{\infty}{i^{-\beta}}$. Set
$\epsilon_k = \frac{\eps}{ \zeta(\beta)k^\beta}$. The algorithm
\textsc{AllWindowSum} is shown as Algorithm~\ref{alg:all-winsum}. Proof of theorem below is 
analogous to Theorem~\ref{thm:winsum}. 

\begin{algorithm}[t]
  \caption{\textsc{AllWindowSum}} \label{alg:all-winsum}
{%\fontsize{10}{10}\selectfont
  \begin{algorithmic}
    \STATE Initialize $\tree = \tree(1, 1)$, with $c_{1,1}$ initialized
    to $\Lap(1/\eps_1)$. 
    \FORALL{updates $x_i$}
    \IF{the rightmost leaf of $\tree$ is $i-1$}
    \STATE   Grow $\tree$ so that $\tree = \tree(1, 2(i-1))$, adding additional
    nodes and variables as necessary; initialize new variables at level
    $k$ to $\Lap(1/\eps_k)$. 
    \STATE Add the value $c^{0}_{1, i-1}$ to the root variable $c_{1, 2(i-1)}$, where $c^{0}_{lu}$ is
    the value of $c_{lu}$ without the Laplace noise.
    \ENDIF
    \STATE Add $x_i$ to all $c_{lu}$ in such that $i \in [l, u]$.
    \STATE Let $W' = 2^{\lceil \log W \rceil}$. At time step $j$,
    output $\hat{F'}_w(j, W) = s((k-1)W, \tree_{k-1}) - s(j -
    W,\tree_{k-1}) +  s(j, \tree_k)$,
    where $k = \lceil j/W' \rceil$.
    \ENDFOR
  \end{algorithmic}}
\end{algorithm}

\begin{theorem}
  \label{thm:all-winsum}
  There exists a constant $K$ s.t.~\textsc{AllWindowSum} satisfies
  $\eps$-differential privacy and achieves $(\delta, \gamma)$-utility,
  where
  \begin{equation*}
    \delta =
    \begin{cases}
      O(\frac{1}{\eps}\log^{1.5\beta} W \log^{0.5} \frac{1}{\gamma}), &\log W \geq
      K \log \frac{1}{\gamma}\\
      O(\frac{1}{\eps}\log^\beta W \log \frac{1}{\gamma}), &\log W < K \log \frac{1}{\gamma}
    \end{cases}
  \end{equation*}
  Furthermore, the algorithm can be implemented to use $O(T)$ words of
  space and run in $O(\log T)$ time per update on inputs consisting of
  $T$ updates.
\end{theorem}
%\junk{
\begin{proof}
  \textbf{Privacy.} The proof of privacy is analogous to the proof of
  privacy for Theorem~\ref{thm:winsum}, but we treat different levels
  of $\tree$ separately and use Theorem~\ref{thm:composition} to bound
  the total privacy loss. More precisely, we show that level $k$ in
  the tree satisfies $\eps_k$-differential privacy and use the fact
  that $\sum_{k=1}^\infty{\eps_k} = \eps$.
  
  \textbf{Utility.} The utility analysis is also analogous to the
  proof of Theorem~\ref{thm:winsum}, noticing the following facts:
  \textbf{(1)} $W \leq W' \leq 2W$; \textbf{(2)} as an upper bound on
  the variance of any variable used to compute $\hat{F'}_w(j, W)$ we
  can use the variance of variables at level $\log W' + 1$, which is
  $O(\log^\beta W)$.
  The rest of the proof is unchanged.
\end{proof}
%}

\subsection{Exponential Decay}

\junk{For the exponential decay sum problem, our goal is to design a
differentially private estimator that outputs an estimate at every
time step, and for any fixed time step $j$ has error at most
$\mathcal{E}$ with constant probability, where $\mathcal{E}$
satisfies: \textbf{(1)} $\mathcal{E} \ll
\frac{1}{1-\alpha}$($\frac{1}{1-\alpha}$ is the size of the range for
the exponential decay sum); \textbf{(2)} $\mathcal{E}$ is independent
of $j$.These properties are analogous to the properties we required
for the window sum problem. We are interested in the regime where
$\alpha \rightarrow 1$, as we want an algorithm whose error bound
grows slowly when the range of the exponential decay sum grows.}

While for the window sum problem we keep a sequence of dyadic trees,
for the exponential decay problem we keep a single dyadic tree that
grows over time. The main property of exponentially decaying sums
that we use is that if $S_1$ is the exponential decay sum over a time
interval $[a, b-1]$ and $S_2$ is the exponential decay sum over a time
interval $[b, c]$, then $\alpha^{c - b + 1}S_1 + S_2$ is the
exponential decay sum over the time interval $[a, c]$. Thus at a node
in the dyadic tree that is indexed by interval $[l, u]$ we can keep
the exponential decay sum for that interval. However, doing this for
every interval results in a data structure with unbounded
sensitivity. We update only the left nodes in the tree and show that
we can bound the sensitivity in that case.

\begin{algorithm}[t]
  \caption{\textsc{ExponentialSum}} \label{alg:expsum}
{%\fontsize{10}{10}\selectfont
  \begin{algorithmic}
  \STATE Set $\lambda =  \frac{1}{\alpha \ln 2}\left(\ln
  \frac{2\alpha}{1 - 1\alpha} + \frac{1}{2} + \ln 2\right)$. 

\STATE Initialize $\tree = \tree(1, 1)$, with $c_{1,1}$ initialized
  to $\Lap(\lambda/\eps)$

  \FORALL{updates $x_i$}
  \IF{the rightmost leaf of $\tree$ is $i-1$}
  \STATE  Grow $\tree$ so that $\tree = \tree(1, 2(i-1))$, adding additional
  nodes and variables as necessary and initializing new variables to
  $\Lap(\lambda/\eps)$. 
  \STATE Add the value $\alpha^{i-1}c^{0}_{1,
    i-1}$ to the root variable $c_{1, 2(i-1)}$, where $c^{0}_{lu}$ is
  the value of $c_{lu}$ without the Laplace noise.
  \ENDIF
  \FORALL{$[l, u]$ such $i \in [l, u]$ and the node indexed by $[l,
    u]$ is a left node}
  \STATE add $x_i\alpha^{u-i}$ to $c_{lu}$
  \ENDFOR
  \STATE output $\hat{F}_e(j, \alpha) = \sum_{k = 0}^r{c_{u_{k},u_{k+1}}\alpha^{j - u_{k+1}}}$.
  \ENDFOR
  \end{algorithmic}}
\end{algorithm}

The \textsc{ExponentialSum} algorithm is shown as
Algorithm~\ref{alg:expsum}. We analyze the algorithm for $\alpha \in
(2/3, 1)$; observe that when $\alpha \leq 2/3$, the range of the $F_e$
is $[0, 3]$, and, thereofore, achieving $(1.5, 0)$-utility is
trivial. Thus $\alpha \rightarrow 1$ is the interesting regime for
approximating $F_e$. 

The following lemma is useful in the analysis.  
\begin{lemma}
  \label{lm:decr-rank}
  For an arbitrary $i$, let $[l_1, u_1], [l_2, u_2], \ldots$ be the
  sequence of intervals such that $\forall k: i \in [l_k, u_k]$ and
  $[l_k, u_k]$ is a left node. Assume the intervals are ordered in
  ascending order of $u_k - l_k$. Then $u_k - i \geq 2^{k-1}-1$.
\end{lemma}
\begin{proof}
  By induction. The base
case is trivial, as from $i \in [l_1, u_1]$ follows $u_1 -i \geq
0$. For the inductive step, it suffices to show that $u_k - u_{k-1}
\geq 2^{k-2}$. By the construction of $\tree$, all nodes indexed by
intervals $[l, u]$ such that $i \in [l, u]$ lie on the path from the
leaf indexed by $i$ to the root of $\tree$. Therefore, all nodes
indexed by $[l_k, u_k]$ for some $k$ are ancestors of $i$, and, by the
construction of $\tree$ we have $u_k - l_k + 1\geq 2^{k-1}$; in
particular, $[l_k, u_k]$ is an ancestor of $[l_{k-1}, u_{k-1}]$ and
$u_{k-1}- l_{k-1} + 1 \geq 2^{k-2}$. By assumption, all nodes indexed
by $[l_k, u_k]$ are left nodes; let the right sibling of $[l_{k-1},
u_{k-1}]$ be the node indexed by $[l'_{k-1}, u'_{k-1}]$. By
construction, $u'_{k-1} - l'_{k-1} = u_{k-1} - l_{k-1}$ and the parent
of both nodes is indexed by $[l_{k-1}, u'_{k-1}]$. All ancestors of
$[l_{k-1}, u_{k-1}]$ are indexed by intervals that contain $[l_{k-1},
u'_{k-1}]$ as a subinterval, and, therefore,
\begin{align*}
  u_k &\geq u'_{k-1} = u_{k-1} +(u_{k-1} - l_{k-1} + 1) \\
  &\geq u_{k-1} + 2^{k-2}    
\end{align*}
This completes the inductive step.
\end{proof}

\begin{theorem}
  \label{thm:exp-sum}
  Assume $\alpha \in (2/3, 1)$ and let $K$ be a universal
  constant. \textsc{ExponentialSum} satisfies $\eps$-differential
  privacy and achieves ($\delta, \gamma$)-utility with
  \begin{equation*}
    \delta = \begin{cases}
      O(\frac{1}{\eps}\frac{\alpha}{1-\alpha} \log^{0.5} \frac{1}{\gamma}), &\log \frac{\alpha}{1-\alpha}
    \geq \log \frac{1}{\gamma}\\
    O(\frac{1}{\eps}\log \frac{\alpha}{1-\alpha} \log \frac{1}{\gamma}),
    &\log \frac{\alpha}{1-\alpha} < \log \frac{1}{\gamma}
    \end{cases}
  \end{equation*}
  Furthermore, \textsc{ExponentialSum} can be implemented to use
  $O(\log T)$ words of space and to run in $O(\log T)$ time per update on
  inputs consisting of $T$ updates.
\end{theorem}
\begin{proof}
  \noindent\textbf{Privacy.} 
It is sufficient to fix $T$ and argue that
\textsc{ExponentialSum} is $\eps$-differentially private for inputs
of size $T$ when all outputs for $j \leq T$ are produced at step $T$.

We analyze the sensitivity of $\tree$. Define $\vec{c}_0(\vec{x})$ as
in the proof of Theorem~\ref{thm:winsum} and $[l_1, u_1], [l_2, u_2],
\ldots$ as in Lemma~\ref{lm:decr-rank}. We have
\begin{align}
  \|\vec{c}_0(x_1, \ldots, x_i, \ldots, x_T) - \vec{c}_0(x_1, \ldots,
  1-x_i, \ldots, x_T)\|_1 &\leq \sum_{k =
    1}^\infty{\alpha^{u_k - i}x_i} \leq \sum_{k = 1}^\infty{\alpha^{u_k - i}}\notag\\
  &\leq \sum_{k = 1}^\infty{\alpha^{2^{k-1} - 1}} =
  \frac{1}{\alpha}\sum_{k =
    0}^\infty{\alpha^{2^k}}\notag\\
  &\leq \frac{1}{\alpha} +  \frac{1}{\alpha}\int_0^\infty{\alpha^{2^x}
    dx}\notag\\
  &= \frac{1}{\alpha} +\frac{1}{\alpha\ln
    2}\int_{\ln\frac{1}{\alpha}}^\infty{\frac{e^{-t}}{t}dt} \notag\\   
  &= \frac{\ln 2 + E_1(\ln{\frac{1}{\alpha}})}{\alpha \ln
    2}.\label{eq:sens-e1}
\end{align}
Here $E_1(x) = E_1(x) = \int_x^\infty{\frac{e^{-t}}{t}dt}$. We have
the following series expansion for $E_1$, which converges for all real
$|x| \leq \pi$~\cite{abramowitz1964handbook}:
\begin{equation}
  \label{eq:e1-series}
  E_1(x) = -\eta - \ln x + \sum_{k = 1}^\infty{\frac{(-1)^{k+1}x^k}{k!k}},
\end{equation}
where $\eta$ is the Euler-Mascheroni constant. Since, by
assumption, $\alpha > e^{-1}$, we have $\ln \frac{1}{\alpha} <
1$. For $x < 1$, the last term in (\ref{eq:e1-series}) is bounded by
$\eta + E_1(1) = \eta + \frac{1}{2}$. Therefore, we have,
\begin{align}
  E_1(2\ln{\frac{1}{\alpha}}) &\leq -\ln \ln \frac{1}{\alpha} +
  \frac{1}{2}\notag\\
  &= \ln \frac{1}{\ln \frac{1}{\alpha}} + \frac{1}{2}\label{eq:e1-log}
\end{align}
For $x \in (0, 2)$, we have the following series expansion for $\ln
x$:
\begin{equation}
  \label{eq:log-series}
  \ln x = x-1 - \sum_{k = 2}^{\infty}{\frac{(1-x)^k}{k}}.    
\end{equation}
Since by assumption $ 1/\alpha - 1 < 1/2$, we have $\ln (1/\alpha)
\geq (1/\alpha - 1)/2$. Substituting in (\ref{eq:e1-log}), we get
\begin{equation}
  \label{eq:e1-final}
  E_1(\ln{\frac{1}{\alpha}}) \leq \ln \frac{1}{\frac{1 - 
      \alpha}{2\alpha}} + \frac{1}{2} = \ln \frac{2\alpha}{1 -
    \alpha} + \frac{1}{2}
\end{equation}
Substituting (\ref{eq:e1-final}) into (\ref{eq:sens-e1}) gives us
the following bound on sensitivity:
\begin{align}
  \|\vec{c}_0(x_1, \ldots, x_i, \ldots, x_T) - \vec{c}_0(x_1&,
  \ldots, 1-x_i, \ldots, x_T)\|_1 \notag\\
  &\leq \frac{1}{\alpha \ln 2}\left(\ln
  \frac{2\alpha}{1 - 1\alpha} + \frac{1}{2} + \ln 2\right)  \label{eq:sens-final}
\end{align}
By Theorem~\ref{thm:laplace} and (\ref{eq:sens-final}),
\textsc{ExponentialSum} satisfies $\eps$-differential privacy.

\medskip
\noindent\textbf{Accuracy.} Clearly, $\E \hat{F}_e(j, \alpha) = F_e(j,
\alpha)$. Next we upper bound $\sigma^2$, the maximum variance
of $\hat{F}(j, \alpha)$ over all $j$. By Lemma~\ref{lm:decr-rank},
all intervals $[1, u_1], [u_1, u_2], \ldots, [u_r, j]$ correspond to
nodes in distinct levels of $\tree$, and therefore have sizes which
are distinct powers of 2. We have, for some fixed constant $C$, 
\begin{align*}
  \sigma^2 &\leq \left(C\frac{\log \frac{\alpha}{1 - \alpha}}{\alpha
      \eps}\right)^2\sum_{i = 1}^\infty{\alpha^{2(2^i - 1)}}\\
  &= \left(C\frac{\log \frac{\alpha}{1 - \alpha}}{\alpha
      \eps}\right)^2\frac{1}{\alpha^2}\sum_{i = 2}^\infty{\alpha^{2^{i}}}\\
  &\leq\frac{1}{\alpha^2} \left(C\frac{\log \frac{\alpha}{1 - \alpha}}{\alpha
      \eps}\right)^3.
\end{align*}
The proof can be completed analogously to the proof of
Theorem~\ref{thm:winsum}.
\end{proof}

\subsection{Polynomial Decay}

Unlike the running sum, window sum, or exponential decay sum problems,
there is no easy way to combine a polynomial decay sums over intervals
$[a, b-1]$ and $[b, c]$ into a polynomial decay sum over $[a,
c]$. Therefore, our techniques for estimating polynomial decay sum are
considerably different. On a high level, we approximate the polynomial
decay function $g(i) = (i+1)^{-c}$ by a function $g'$ that is constant
on exponentially growing in size intervals. Then we can approximate
the decay sum induced by $g'$ by running multiple instances of our
window sum algorithm in parallel. This technique results in a
bi-criteria approximation, i.e.~our approximation guarantee has both a
multiplicative and an additive approximation factor. As $c \rightarrow
1$ (i.e.~as the range of the polynomial decay sum grows), the additive
approximation factor remains bounded and is dominated by $\beta^{-2}$,
where $(1\pm\beta)$ is the multiplicative approximation factor. Thus
the approximation guarantees for our algorithm are mostly determined
by a 
%user-chosen 
trade-off between additive and multiplicative
approximation. 
%The technique is more general than our techniques for
% window sum and exponential decay sum.

For a given polynomial decay function $g = (i + 1)^{-c}$ and the
induced decay sum $F$, let us a fix
a multiplicative error parameter $\beta$ and define a function $b$ as
$\forall j \geq 1: b(j) = \max\{i: g(i) \geq
(1-\beta)^j\}$ and $b(0) = 0$. Intuitively $g(i)$ is almost constant for $i \in
[b(j-1), b(j))$. 

We can now define a function $g'$ that approximates $g$:
%\begin{align}
  $\forall i \in [b(j-1), b(j)) \junk{\cap [0, i^*)}: g'(i) =
  (1-\beta)^{j-1}$
  \junk{&\forall i \geq i^*: g'(i) = 0.}
%\end{align}
Let $F'$ be the decay sum induced by $g'$. From the definition of $g'$
it is immediate that
%\begin{equation}
 % \label{eq:F'-\}
  $\forall j, \forall \vec{x} \in \{0, 1\}^j:  (1-\beta) F(j) \leq
  F'(j) \leq F(j).$
%\end{equation}

\begin{algorithm}[t]
  \caption{\textsc{PolynomialSum}} \label{alg:poly}
{%\fontsize{10}{10}\selectfont
  \begin{algorithmic}
    \STATE Set $\lambda = \frac{\log (1/(1-\beta))}{c \beta^2}
    +\frac{1}{\beta}$. 
    \STATE Start an instance of \textsc{WindowSum} for input $x_1,
    \ldots$ with window size $W_1 = b(1)$ and initializing noise
    for each variable $\Lap(\lambda/\eps)$. Set $j^* = 1$.
    
    \FORALL{updates $x_i$}
    \IF{$i = b(j^*) + 1$}
    \STATE Set $j^* = j^* + 1$.
    \STATE Start a new instance of
    \textsc{WindowSum} with window size $W_{j^*} =
    b(j^*) - b(j^*-1)$ and and initializing noise for each variable
    $\Lap(\lambda/\eps)$.
    \ENDIF
    \FORALL{$k \leq j^*$}
    \STATE Update the $k$-th instance of \textsc{WindowSum} with input $(1-\beta)^{k - 1}x_{i - b(k-1)}$
    \ENDFOR
    \STATE Output $\hat{F}_p(i, c) = \sum_{j\geq 0: b(j) < i}{F_w((1-\beta)^jx_{1}, \ldots, (1-\beta)^jx_{i-b(j)}, W_{j+1})}$. 
    \ENDFOR
  \end{algorithmic}}
\end{algorithm}

The \textsc{PolynomialSum} algorithm is shown as
Algorithm~\ref{alg:poly}. Note that we call the $j$-th instance of
\textsc{WindowSum} with input consisting of time updates in $\{0,
(1-\beta)^{j-1}\}$. It is straightforward to check that the
\textsc{WindowSum} algorithm can handle such scaled instances without
modification. Note also that we modify the \textsc{WindowSum}
algorithm slightly by adjusting the magnitude of noise added to the
variables associated with the dyadic trees kept by \textsc{WindowSum}.

\begin{theorem}\label{thm:poly-ub}
  \textsc{PolynomialSum} satisfies $\eps$-differential
  privacy, and for any  $j$, with probability $1 - \gamma$, we have $(1-\beta)F_p(j, c) - O(\delta) \leq \hat{F}_p(c) \leq F_p(j, c) +   O(\delta)$, where
  \begin{equation*}
    \delta =
    \begin{cases}
      \frac{1}{\eps}\left(\frac{1}{c\beta^2}\log\frac{1}{1-\beta}\right)^{1.5}\log^{0.5}
      \frac{1}{\gamma} &\text{ if }
      \frac{1}{c\beta^2}\log\frac{1}{1-\beta} \geq \log
      \frac{1}{\gamma}\\
      \frac{1}{\eps}\frac{1}{c\beta^2}\log\frac{1}{1-\beta}\log
      \frac{1}{\gamma} &\text{ if }
      \frac{1}{c\beta^2}\log\frac{1}{1-\beta} < \log \frac{1}{\gamma}
    \end{cases}
  \end{equation*}
  Furthermore, \textsc{PolynomialSum} can be implemented to use $O(T)$
  words of space and run in $O(\log^2 T/\log(1/(1-\beta)))$ time per
  update on inputs consisting of $T$ updates.
\end{theorem}
\begin{proof}
  \textbf{Privacy.} The privacy analysis is analogous to the analysis in
the proof of Theorem~\ref{thm:winsum}, but we bound sensitivity over
all instances of \textsc{WindowSum}. Due to the scaling of the input,
the sensitivity of the $j$-th instance of \textsc{WindowSum} is
bounded by $(1+\beta)^{j-1}(\log W_j+1)$. Let us first bound
$W_j$. Observe that $b(j) = \lfloor g^{-1}((1-\beta)^j \rfloor$. For
$g(i) = (i+1)^{-c}$, we have $b(j) \in [(1-\beta)^{-j/c} - 2,
(1-\beta)^{-j/c} - 1]$. Then $W_j$ can be bounded as $W_j = b(j) -
b(j-1) \leq (1-\beta)^{-j/c} - (1-\beta)^{-(j-1)/c} + 1$.  Since
$1-\beta < 1$ and $j \geq 1$, we have $W_j \leq (1-\beta)^{-j/c}$. We
can then bound the overall sensitivity is by
  \begin{align}
    \sum_{j = 1}^\infty{(1-\beta)^{j-1}\log W_j} +
    \sum_{j=0}^\infty{(1-\beta)^j} &\leq \sum_{j =
      1}^\infty{(1-\beta)^{j-1} \log
      \frac{1}{(1-\beta)^{j/c}}} + \frac{1}{\beta}\notag\\
    &= \frac{1}{c}\log \frac{1}{(1-\beta)}\sum_{j =
      1}^\infty{j(1-\beta)^{j-1} } + \frac{1}{\beta}\notag\\
    &= \frac{1}{c\beta^2}\log \frac{1}{(1-\beta)} +
    \frac{1}{\beta}\label{eq:poly-sens}
  \end{align}
  Theorem~\ref{thm:laplace} and (\ref{eq:poly-sens}) complete the
  privacy proof.

  \textbf{Accuracy.} Note that $\E \hat{F}_p(j, c) = F'(j)$. The
  variance of $F_w((1-\beta)^jx_{b(j)}, \ldots, (1-\beta)^jx_k, W_j)$
  is at most $2(1-\beta)^{2j}\lambda^2\log W_j$. Therefore, the total
  variance $\sigma^2$ of $\hat{F}_p(j, c)$ is
  \begin{align*}
    \sigma^2 &\leq\lambda^2
    \frac{1}{c}\log\frac{1}{1-\beta}\sum_{j=0}^\infty{(j+1)(1-\beta)^{2j}}\\
    &= \lambda^2\frac{1}{c\beta^2(2-\beta)^2}\log\frac{1}{1-\beta}\\
    &= O\left(\left(\frac{1}{c\beta^2}\log\frac{1}{1-\beta}\right)^3\right)
  \end{align*}
  Using Lemma~\ref{lm:chernoff} as in Theorem~\ref{thm:winsum} we can
  show that for any $j$, with probability at least $1-\gamma$,
  \begin{equation*}
    |\hat{F}_p(j, c) - F'(j)| = 
    \begin{cases}
      O((\frac{1}{c\beta^2}\log\frac{1}{1-\beta})^{1.5}\log^{0.5}
      \frac{1}{\gamma}) &\text{ if }
      \frac{1}{c\beta^2}\log\frac{1}{1-\beta} \geq \log
      \frac{1}{\gamma}\\
      O(\frac{1}{c\beta^2}\log\frac{1}{1-\beta}\log \frac{1}{\gamma})
      &\text{ if } \frac{1}{c\beta^2}\log\frac{1}{1-\beta} < \log
      \frac{1}{\gamma}
    \end{cases}
  \end{equation*}
  Since for all $\vec{x}$ and all $j$, $(1-\beta)F(j) \leq F'(j) \leq
  F(j)$, this completes the proof. 
\end{proof}
\junk{
\begin{proof}
  \textbf{Privacy.} The privacy analysis is analogous to the analysis
  in the proof of Theorem~\ref{thm:winsum}, but we bound sensitivity
  over all instances of \textsc{WindowSum}. 

  \textbf{Accuracy.} Note that $\E \hat{F}_p(j, c) = F'(j)$. The proof
  proceeds by bounding the total variance $\sigma^2$ of $\hat{F}_p(j,
  c)$ and using the variance bound to bound $|\hat{F}_p(j, c) -
  F'(j)|$. Applying the property of $g'$
  %Applying (\ref{eq:F'-apx-F}) 
  completes the
  argument. Intuitively, the variance can be bounded because instances
  of \textsc{WindowSum} that are run with a polynomially larger window
  parameter are given exponentially smaller updates, and the
  exponential factor dominates.

  With respect to space complexity, we use the fact that each instance
  of \textsc{WindowSum} has space complexity linear in the window
  size, and the windows of all running instances partition the
  input. The bound on the running time per update is somewhat loose
  and uses the fact that at each time step we update at most $c \log
  T/\log(1/(1-\beta))$ instances of \textsc{WindowSum}, and each
  instance takes at most $\log T$ time per update.

  We give full details in Appendix~\ref{app:poly-ub}. 
\end{proof}
}

This algorithm can more generally be used to compute a private (under
continual observation) approximation to a decayed sum $F$ induced by a
decay function $g$ as long as $g^{-1}$ grows subexponentially. In this
case sensitivity remains bounded and the additive error guarantee is
dominated by a function of $\beta$, but the exact function depends on
$g$. The algorithm is not applicable to the window or running sum
problem, since for them $g^{-1}$ is not well defined; the guarantee
for exponential decay sum is incomparable with the one in
Theorem~\ref{thm:exp-sum}.

\section{Lower Bounds}

We give a general framework for lower bounding the dependence of the error
$\delta$ on the error probability $\gamma$ for algorithms that are private under continual
observation and achieve $(\delta, \gamma)$-utility. \junk{This framework is
  inspired by work on lower bounding the approximation factor of
  differentially private combinatorial optimization
  algorithms~\cite{gupta} and lower bounds proved
  in~\cite{dwork-continual} fit within the framework. To the best of
  our knowledge, this framework has not appeared explicitly in the
  literature before.} We also instantiate the framework with a
construction that yields concrete lower bounds for the three decay sum
problems considered in this paper. As far as the dependence on error
probability is concerned, our lower bounds for window and exponential
decay sums are tight. Our
lower bound for polynomial decay sums is against a purely additive
approximation and is not directly comparable to the upper bounds on the
approximation factors of our algorithm. \junk{These concrete lower
  bounds are an extension of the lower bound of $\log T$ (where $T$ is
  a bound on the size of the input) proved
  in~\cite{dwork-continual}. We extend their lower bound in two ways:
  \textbf{(1)} we show lower bounds for decay sum problems that have
  not been considered before; \textbf{(2)} we show that the lower
  bound is really on the dependence of the additive error on the
  number of estimates $q$ required to be simultaneously accurate,
  rather than on the size of the input; this brings out fact that the
  dependence on $q$ is optimal and extends to algorithms that are not
  required to be accurate at every time step.}

Suppose that for a fixed error probability $\gamma$, we want to prove
a lower bound on $\delta$ for any $\eps$-differentially private
algorithm that achieves $(\delta, \gamma)$-utility with respect to a
function $F(x_1, \ldots, x_j)$ We can take $\gamma = 2/3q$, and it
follows, by the union bound, that for any set $Q \subseteq [n]$ of size $q$, with
probability $2/3$, the algorithm is within an absolute error $\delta$
from $F(x_1, \ldots, x_j)$ for all $j \in Q$. Assume that for some $T$
we can construct $N+1$ instances $\vec{x}^0, \ldots, \vec{x}^N$, each
of length $T$, that satisfy the following properties:
\begin{packed_enum}
\item \emph{$(Q, \delta)$-independence}: for all $a, b \in \{0, \ldots, N\}, a\neq
  b$, there exists some $j \in Q \subseteq T$ such that $|F(x^a_1, \ldots, x^a_j)
  - F(x^b_1, \ldots, x^b_j)| > 2\delta$.

\item \emph{$D$-closeness}: for all $a, b \in \{0, \ldots, N\}$, we have
  $d_H(\vec{x}^a, \vec{x}^b) \leq D$, where $d_H$ is the standard
  Hamming distance.
\end{packed_enum}

\begin{lemma}
  \label{lm:lb-framework}
  Assume there exists an $\eps$-differentially private algorithm
  $\Alg$ that at time step $j$ outputs $\hat{F}(x_1, \ldots,
  x_j)$. Assume further that for any $Q \subseteq \mathbb{N}$, $|Q| =
  q$, we have  
  %\begin{equation*}
    $\Pr[\forall j \in Q: |\hat{F}(x_1,\ldots, x_j) - F(x_1,\ldots,
    x_j)| \leq \delta] \geq 2/3.$
  %\end{equation*}
  If for some $Q$ there exists a set $\{\vec{x}^0, \ldots,
  \vec{x}^N\}$ that satisfies $(Q, \delta)$-independence and $D$-closeness
  with respect to $F$, then
  %\begin{equation*}
    $D > \frac{\ln N + \ln 2}{\epsilon}$
  %\end{equation*}
\end{lemma}
\begin{proof}
    Let $B(\vec{x}^i) = \{\vec{f}: |f_j - F(x^i_1, \ldots, x^i_j)| \leq
  \delta\}$. By assumption, $\Pr[(\hat{F}(x^i_1, \ldots, x^i_j))_{j =
    1}^T \in B(\vec{x}^i)] \geq 2/3$. Then, by the definition of
  differential privacy and $D$-closeness, we have
  \begin{equation*}
    \forall i: \Pr[(\hat{F}(x^0_1, \ldots, x^0_j))_{j = 1}^T \in B(\vec{x}^i)]
    \geq e^{-\eps D}2/3.
  \end{equation*}
  By $(Q, \delta)$-independence, $B(\vec{x}^a) \cap B(\vec{x}^b) =
  \emptyset$ for all $a \neq b$. Therefore,
  \begin{equation*}
    \Pr[(\hat{F}(x^0_1, \ldots, x^0_j))_{j = 1}^T \in
    \bigcup_{i=1}^N{B(\vec{x}^i)}] = \sum_{i=1}^N{\Pr[(\hat{F}(x^0_1,
      \ldots, x^0_j))_{j = 1}^T \in B(\vec{x}^i)]} \geq Ne^{-\eps D}2/3.
  \end{equation*}
  However, since $ B(\vec{x}^0) \cap \bigcup_{i=1}^N{B(\vec{x}^i)} =
  \emptyset$, by the assumptions on $\Alg$ we have
  \begin{equation*}
    \Pr[(\hat{F}(x^0_1, \ldots, x^0_j))_{j = 1}^T \in
    \bigcup_{i=1}^N{B(\vec{x}^i)}] < 1/3.
  \end{equation*}
  Therefore, $2N  < e^{\eps D}$, and the lemma follows by taking logarithms.
\end{proof}

In order to apply Lemma~\ref{lm:lb-framework}, we need a method to
construct a set of instances satisfying $(Q, \delta)$-independence and
$D$-closeness for a given error bound $\delta$, such that $D$ is upper
bounded by a function of $\delta$ and $N$ is lowerbounded by a function of
$|Q|$. We show a construction that allows us to derive a lower bound
for \emph{any decayed sum problem}, where, naturally, the form of the
lower bound depends on the specific problem, i.e.~on the decay
function $g$. As corollaries, we derive specific lower bounds for the
problems we consider in this paper. In our construction, the set of
vectors $\{\vec{x}^i\}_{i = 0}^q$ is defined as $\vec{x}^0 =
(0^{Dq})$ and $\vec{x}^i = (0^{(i-1)D}, 1^{D}, 0^{(q-i)D})$. We set
$Q = \{j: D \text{ divides } j\}$ and choose $\delta$ according to the
specific decay function $g$. Consider a general decayed sum function
$F(x_1, \ldots, x_j)$ with a decay function $g$. The construction
gives our main lower bound theorem.

\begin{theorem}
  \label{thm:lb-main}
  Assume there exists an $\eps$-differentially private algorithm
  $\Alg$ that at time step $j$ outputs $\hat{F}(x_1, \ldots,
  x_j)$ and achieves $(\delta, \gamma)$-utility with respect to a
  decayed sum function $F$ induced by $g$. Denote $G(x) = \sum_{i
    = 0}^{x-1}{g(i)}$. Then\junk{, for $\log q/\eps = O(\alpha/(1-\alpha))$}
  %\begin{equation*}
    $\delta \geq \frac{1}{2}G(\Omega(\frac{\log (1/\gamma)}{\eps})).$
  %\end{equation*}
\end{theorem}

For the three problems considered in this paper we derive the
following corollaries.

\begin{corollary}
  Assume there exists an $\eps$-differentially private algorithm
  $\Alg$ that at time step $j$ outputs $\hat{F}_w(j, W)$ and achieves
  $(\delta, \gamma)$-utility with respect to $F_w(j, W)$.  Then,
  %\begin{equation*}
    $\delta \geq \Omega\left(\min\left\{ \frac{W}{2}, \frac{\log (1/\gamma)}{\eps}
      \right\}\right).$
  %\end{equation*}
\end{corollary}
Note that the lower bound of~\cite{dwork-continual} is a special case
of the above corollary for $\gamma = 2/3W = 2/3T$.

\begin{corollary}
  Assume there exists an $\eps$-differentially private algorithm
  $\Alg$ that at time step $j$ outputs $\hat{F}_e(j, \alpha)$ and
  achieves $(\delta, \gamma)$-utility with respect to $F_e(j, \alpha)$. 
  Then,  for $\alpha \in (2/3, 1)$ we have
  %\begin{equation*}
    $\delta \geq \Omega\left(\min\left\{\frac{\alpha}{1-\alpha}, \frac{\log
          (1/\gamma)}{\eps}\right\}\right).$
  %\end{equation*}
\end{corollary}

\begin{corollary}
  Assume there exists an $\eps$-differentially private algorithm
  $\Alg$ that at time step $j$ outputs $\hat{F}_p(j, c)$ and achieves
  $(\delta, \gamma)$-utility with respect to $F_p(j, c)$.   Then, 
  %\begin{equation*}
    $\delta \geq H_c(\Omega(\frac{\log (1/\gamma)}{\eps})) \geq
    \Omega\left(1 - \frac{\eps^{c-1}}{\log^{c-1} (1/\gamma)}\right),$
  %\end{equation*}
  where $H_c(k)$ is the $k$-th generalized harmonic number in power
  $c$. 
\end{corollary}

\junk{
such an application of the lemma for the window
sum problem.

Consider a set $\{\vec{x}^i\}_{i = 0}^q$ defined as follows:
\begin{align}
  \vec{x}^0 &= (0^{(2\delta + 1)q}\\
  \vec{x}^i &= (0^{(2\delta + 1)(i-1)}, 1^{2\delta+1}, 0^{(2\delta + 1)(q-i)}).
\end{align}
Assume $\delta > \lceil W/2\rceil$. It is easy to verify that this set
satisfies $(Q, \delta)$-independence and $2\delta+1$-closeness with respect to
$F_w(\cdot, W)$, where $Q = \{j: (2\delta+1) \text{ divides } j\}$ (and,
therefore, $|Q| = q$). Given this construction, the following theorem
is an immediate consequence of Lemma~\ref{lm:lb-framework}.

\begin{theorem}
  Assume there exists an $\eps$-differentially private algorithm
  $\Alg$ that at time step $j$ outputs $\hat{F}_w(j, W)$. Assume
  further that for any $Q \subseteq \mathbb{N}$, we have
  \begin{equation*}
    \Pr[\forall j \in Q: |\hat{F}_w(j, W) - F(j, W)| \leq \delta] \geq 2/3.
  \end{equation*}
  Then, 
  \begin{equation*}
    \delta \geq \min{\left\{\frac{W}{2}, \frac{\ln |Q| - \ln 2}{2\epsilon} -
          \frac{1}{2}\right\}}. 
  \end{equation*}
\end{theorem}
}

\section{Extensions and Applications}
\label{app:ext}

Algorithms for sum problems can be used to compute 
more sophisticated statistics as we described earlier. In this section we exhibit a few
extensions and applications of our algorithms. We show how they can be
used to compute sums over individual predicates and some special cases
of sums over holistic predicates, including distinct
counts which is of great interest.  We also show how to compute histograms (over windows
or decayed). In the following discussion we denote an arbitrary
universe as $\uni$.

\subsection{Individual Predicates}

We define an \emph{individual predicate} abstractly as a function ${\cal P}:\uni
\rightarrow [0, 1]$. \junk{These predicates can  encode a
sophisticated computation. For example, in a recommendation system a
predicate can encode the event of a user following a recommendation
(predicate has value 1) or not following (predicate has value 0). A
window sum over such a predicate would then reveal in how many of the
last $W$ events the user has followed the recommendation.
}
Let the input at time step $i$ be $u_i$, where $u_i \in \uni$. The
\emph{decayed predicate sum} for an individual predicate ${\cal P}$ and
decayed sum function $F$ then is $F({\cal P}(u_1), \ldots,
{\cal P}(u_j))$. Differential privacy and utility for predicate sums can be
defined analogously to decayed sums. The following claim is immediate
for individual predicates:

\begin{theorem}
  Let $\Alg$ be an $\eps$-differentially private algorithm that
  achieves $(\delta,\gamma)$-utility with respect to a decayed sum
  $F$. Then, on input ${\cal P}(u_1), \ldots, {\cal P}(u_T)$, $\Alg$ is
  $\eps$-differentially private with respect to  $u_1, \ldots,
  u_T$ and and achieves $(\delta, \gamma)$-utility with respect to the
  decayed predicate sum for ${\cal P}$ and $F$.
\end{theorem}

\subsection{Holistic Predicate Sum}

Individual predicates are limited in that they can depend only on a
single update $u_i$ rather than the whole sequence of updates. Here we
define the more general notion of {\em holistic predicates} and treat the
special case of low-sensitivity holistic predicates, with the distinct
 count problem as an important application.

A \emph{holistic predicate} is a function ${\cal P}: \uni^* \rightarrow [0,
1]$. The decayed predicate sum for the holistic predicate ${\cal P}$ is
$F({\cal P}(u_1), \ldots, {\cal P}(u_1, \ldots, u_j))$.

Let us call a holistic predicate \emph{$k$-sensitive} if for any
sequence of updates $u_1, \ldots, u_T$, any $j \in [T]$ and any $u_j'
\in \uni$, the sequences ${\cal P}(u_1, \ldots, u_j), \ldots$, ${\cal
  P}(u_1, \ldots, u_j, \ldots, u_T)$ and ${\cal P}(u_1, \ldots, u'_j),
\ldots, {\cal P}(u_1, \ldots, u'_j, \ldots, u_T)$ differ in at most
$k$ components.  The following theorem follows from the basic
properties of $\eps$-differential privacy (proof omitted).
\begin{theorem}
  \label{thm:holistic-pred}
  Let $\Alg$ be an $\eps$-differentially private algorithm that
  achieves $(\delta,\gamma)$-utility with respect to a decayed sum
  $F$. Then, when given input ${\cal P}(u_1), \ldots$, ${\cal P}(u_1,
  \ldots, u_T)$ for a $k$-sensitive holistic predicate ${\cal P}$,
  $\Alg$ is $k\eps$-differentially private with respect to $u_1,
  \ldots, u_T$ and and achieves $(\delta, \gamma)$-utility with
  respect to the decayed predicate sum for ${\cal P}$ and $F$.
\end{theorem}

We can show that the fundamental \emph{distinct count} problem can
be encoded as a $2$-sensitive holistic predicate. In the distinct
element count problem the input is a sequence of updates $u_1, u_2,
\ldots$, and at each time step $j$ the goal is to approximate the
number of distinct elements seen so far, i.e.~$|\{u \in \uni: \exists i
\leq j \text{ s.t. } u_i = u\}|$. This problem is equivalent to a
predicate sum problem where $F$ is simply the running sum function,
and ${\cal P}(u_1, \ldots, u_j)$ is 0 when $\exists i < j: u_i = u_j$ and 1
otherwise. The proof of the following lemma is deferred to the full
version of the paper.

\begin{lemma}
  The predicate ${\cal P}(u_1, \ldots, u_j) = \mathbf{1}(\not \exists i < j:
  u_i = u_j)$ is $2$-sensitive.
\end{lemma}
\junk{\begin{proof}
  Assume for some $j$ the element $u_j$ is substituted with
  $u_j'$ and denote the corresponding sequences of predicate values by
  $\vec{p}$ and $\vec{p}'$. Then, we have one of three cases:
  \begin{itemize}
  \item ${\cal P}(u_1, \ldots, u_j) = {\cal P}(u_1, \ldots, u'_j)$ in which case
    $\vec{p} = \vec{p}'$;
  \item ${\cal P}(u_1, \ldots, u_j) = 1$ and ${\cal P}(u_1, \ldots, u_j') = 0$: then
    $\vec{p}$ and $\vec{p}'$ differ in the $j$-th and $k$-th
    coordinates only, where $k > j$ is the smallest integer such that
    $u_k = u_j$;
  \item ${\cal P}(u_1, \ldots, u_j) = 0$ and ${\cal P}(u_1, \ldots, u_j') = 1$:
    symmetric with the previous case.
  \end{itemize}
  In any case, we have $|\{i: p_i \neq p'_i\}| \leq 2$. 
\end{proof}}

Then, by Theorem~\ref{thm:holistic-pred} and the algorithm of Dwork et
al.~\cite{dwork-continual} for the running sum problem, we have the
following result:

\begin{theorem}
  There exists an $\eps$-differentially private algorithm that
  achieves $(\delta, \gamma)$-utility for the discrete element count
  problem with $\delta = O(\log^{1.5} T \log^{0.5} \frac{1}{\gamma})$
  (in the case $\log T = \omega(\log \frac{1}{\gamma})$) or $\delta =
  O(\log T \log\frac{1}{\gamma})$ (in the case $\log T = O(\log
  \frac{1}{\gamma})$),  where $T$ is the number of updates.
\end{theorem}

\junk{Using the same predicate with the window sum function gives the number
of distinct elements seen in the last $W$ updates and not seen in the
first $T - W$ updates. Using similar ideas we can also privately
estimate the difference between the number of distinct elements seen
in the last $W$ updates and the number of distinct elements seen in
the previous $W$ updates. }We leave open the problem of
designing a private algorithm for estimating, at each time step, the
number of distinct elements seen over the last $W$ updates, with
absolute error polylogarithmic in $W$.

\subsection{Histograms}

Consider a situation in which each update can belong to one of several
categories. More formally, let the update at time step $i$ be $(u_i,
x_i) \in \uni \times [0, 1]$. Let $\vec{x}(u, j)$ be $\vec{x}$ restricted to all
components $x_i$ for $i \leq j$ such that $u_i = u$. Then, at time
step $j$, the algorithm outputs a vector $\vec{y}(j) \in
\mathbb{R}^\uni$, where $y_u(j)$ is an approximation to $F(\vec{x}(u,
j))$, for some decayed sum function $F$. We call this the
\emph{decayed histogram} problem for $F$. Differential privacy under
continual observation for decayed histogram problems can be defined
analogously to decayed sum problems.

Given an algorithm to approximate a decayed sum, it 
can be easily extended to an algorithm for the corresponding
decayed histogram problem. 
\begin{theorem}
  Let $\Alg$ be an $\eps$-differentially private algorithm that
  achieves $(\delta, \gamma)$-utility with respect to a decayed sum
  $F$. Then, there exists an $\eps$-differentially private
  algorithm $\Alg'$ that uses $\Alg$ as a black box and for each $j$
  and each $u$ satisfies $\Pr[|y_u(j) - F(\vec{x}(u, j))| > \delta] <
  \gamma$. 
\end{theorem}

\section{Conclusion}

We were inspired by the recent work on differential privacy of data
analysis with continual
updates~\cite{dwork-continual,chan2010private}, a research direction
motivated by monitoring applications. However, our observation is that
in monitoring applications typically recent data is more important
than distant data. Hence, we need analyses that are accurate on the
most recent window of data or data where past is decayed (polynomially
or exponentially, as is common in database streaming). 
%This required
%us to reexamine the notion of privacy and observe that we still needed
%to guarantee differential privacy for the entire stream while being
%accurate within the window.

%With this framework, 
We presented upper and lower bounds for a general
class of functions --- predicate sums --- on window and decayed data.
We derived our upper bounds by balancing noise at different levels of
a tree atop the data in a nontrivial way, and derived lower bounds by
inspiration from work on privacy of 
optimization problems. 
%Taken together, these provide a framework for
%differential privacy on windows and decayed streams.

%The database perspective often induces new variations on the theme of
%(differential) privacy and our work here is an example where new
%theory methods are needed for adapting continual differential privacy
%to our perspective of window or decayed data. 
There are many 
analyses of great interest on decayed data with differential privacy that remain open.

%to the database community, including working with
%data deletions as well as other statistical and data mining primitives
%like  clustering or estimating join size on window/decayed data with
%differential privacy.

\bibliographystyle{alpha}
\bibliography{windowcontp}

%\appendixpage
%\newpage
%\appendix

%\newpage

\end{document}